\documentclass[12pt,a4paper]{article}
\usepackage{epsf, graphicx}
\usepackage{latexsym,amsfonts,amsbsy,amssymb}
\usepackage{amsmath,amsthm}
\usepackage{blkarray}
\usepackage{color,xcolor}
\usepackage{cleveref}
\usepackage{indentfirst}
\usepackage{bm}
\usepackage [latin1]{inputenc}
\makeatletter

\@addtoreset{equation}{section} \makeatother
\newtheorem{Theorem}{Theorem}[section]
\newtheorem{Lemma}{Lemma}[section]

\newtheorem{Remark}{Remark}[section]
\newtheorem{Definition}{Definition}[section]
\newtheorem{Proposition}{Proposition}[section]
\newtheorem{Example}{Example}[section]
\makeatletter \@addtoreset{equation}{section} \makeatother
\begin{document}
\title{Construction of non-generalized Reed-solomon MDS codes based on systematic generator matrix}

\author{Shengwei Liu,~Hongwei Liu,~Bocong Chen\textsuperscript{*}}
\date{}
\maketitle
\insert\footins{\small
\noindent Shengwei Liu is with the School of Mathematics and Information Science, Guangxi University, Nanning 530004, China (e-mail: shwliu@gxu.edu.cn).\\
Hongwei Liu is with the School of Mathematics and Statistics and the Key Laboratory NAA-MOE, Central China Normal University, Wuhan 430079,
China (e-mail: hwliu@ccnu.edu.cn).\\
 Bocong Chen is with the School of Mathematics, South China University of Technology, Guangzhou 511442, China (e-mail:
bocongchen@foxmail.com).\\

 }
{\centering\section*{Abstract}}
 \addcontentsline{toc}{section}{\protect Abstract} 
 \setcounter{equation}{0} 
 Maximum distance separable (MDS) codes are considered optimal
 because the minimum distance cannot be improved for a
 given length and code size.
 The most prominent MDS codes are likely the generalized Reed-Solomon (GRS) codes. In 1989, Roth and Lempel constructed a type of MDS code that is not a GRS code (referred to as  non-GRS). In 2017, Beelen et al. introduced  twisted Reed-Solomon (TRS) codes
 and demonstrated that many MDS TRS codes are indeed non-GRS. Following this, the definition of TRS codes was generalized to the most comprehensive form, which we refer to as  generalized twisted Reed-Solomon (GTRS) codes.
In this paper, we prove that two families of GTRS codes are
non-GRS and provide a systematic generator matrix
for a class of GTRS codes. Inspired by the form of
the systematic generator matrix for GTRS codes,
we also present a construction of non-GRS MDS codes.

\medskip
\noindent{\large\bf Keywords:~}\medskip MDS codes, non-GRS MDS codes, Schur square, Cauchy matrix, GTRS codes

\noindent{\bf2010 Mathematics Subject Classification}: 94B05, 94B65.

\section{Introduction}
A linear code $C$ of length $n$, dimension $k$
and minimum Hamming distance $d$ over
the finite field $\mathbb{F}_{q}$ is
referred to as an $[\,n, k, d\,]$ code.
If the parameters of the code $C$ satisfy
 the Singleton bound $d=n-k+1$, then $ C$
 is classified as an MDS code.
The most prominent MDS codes probably are
the GRS codes \cite{R. Roth}.

Determining whether a given infinite family
of MDS codes is GRS based on the definition is
often not straightforward.
For an $[n,k]$ linear code $C$,
a generator matrix  of the form
$(I_{k}|A_{k\times(n-k)})$ is known as its
{\it systematic generator matrix}.
In \cite{G. Seroussi},
the authors proved that an MDS code with a
generator matrix of the form $(I|A)$ is GRS
if and only if $A$ is a {\it Cauchy matrix}.
For convenience, we refer to it as a Cauchy matrix,
as the authors also discussed extended Cauchy matrices.
In \cite{A. Lempel1}, the authors generalized
the form of the matrix $A$ from
 \cite{G. Seroussi} (see  \cite[Lemma 1]{A. Lempel1}).
 Those results from \cite{G. Seroussi} and \cite{A. Lempel1}
 provide a powerful method (referred to as the Cauchy matrix method) for
 determining whether an MDS code is GRS. Furthermore, in \cite{A. Lempel}, the authors constructed a non-GRS MDS code by extending
 a column of the generator matrix of a GRS code
 and proved its non-GRS property using the Cauchy matrix method.

Alternatively, the {\it Schur product method}
is another effective approach for constructing non-GRS codes.
Given vectors $\bm x=(x_{1},\dots,x_{n})
$
and
$\bm y=(y_{1},\dots,y_{n})$ of $\mathbb{F}_{q}^{n}$,
the Schur product of $\bm x$ and $\bm y$ is defined as $\bm x*\bm y=(x_{1}y_{1},\dots,x_{n}y_{n})$. For two linear codes $C,~D\subseteq\mathbb{F}_{q}^{n}$, the {\it product code}
$C*D$ is a linear subspace of $\mathbb{F}_{q}^{n}$
generated by all products $\bm x*\bm y$ where $\bm x\in C,~\bm y\in D$.
If $C=D$, we denote $C*C=C^{2}$ as the {\it square} of $C$.
For the basic results about product of linear codes,
we refer to \cite{H. Randriambololona}.
In \cite{D. Mirandola}, Mirandola and Z\'{e}mor characterized
the codes $C$ and $D$ whose product has the
maximum possible minimum distance,
 showing
 that an $[\,n, k, d\,]$ MDS code with
 $k\leq\frac{n-1}{2}$ is GRS if and only if $\dim(\mathcal{C}^{2})=2k-1$.
 This result also provides a robust method
 (referred to as the Schur product method)
 for determining whether an MDS code is GRS.
 In works \cite{C.ZHU,P. Beelen1,P. Beelen},
 the authors proved that GTRS codes contain many
 non-GRS MDS codes using the Schur product method.
 In \cite{S.LIU}, sufficient and necessary conditions
 for constacyclic codes to be GRS were provided
 through the Schur product method.

In 2017, twisted Reed-Solomon (TRS) codes were
initially introduced by Beelen et al. in \cite{P. Beelen}.
Subsequently, the definition of TRS codes was
generalized to its most comprehensive form,
referred to as generalized twisted
Reed-Solomon (GTRS) codes in this paper
(see \cite{WEIDONG,DING}).
Numerous studies have focused on self-dual
GTRS codes, MDS and near-MDS GTRS codes,
parity-check matrices of GTRS codes, etc.
(see \cite{WEIDONG}-\cite{GU}, \cite{H.LIU},
\cite{YQ}-\cite{C.ZHU1}). A significant
 motivation to study GTRS codes is that many MDS
 GTRS codes are non-GRS (as referenced in
  \cite{P. Beelen1} and \cite{C.ZHU}).
  The Schur product method is the primary
  tool used to prove that many MDS GTRS codes
  are non-GRS. The construction of non-GRS MDS
  codes has garnered interest in recent years.
  Recently, in \cite{chenhao}, the author constructed
  non-GRS MDS codes from arbitrary genus algebraic curves,
  while the authors of \cite{LIF} developed a type of
  non-GRS MDS code based on cyclic codes.
  The non-GRS property of GTRS codes has been
  established for many parameters
  but not all in \cite{C.ZHU,P. Beelen1,P. Beelen}.

 In this paper, we prove that two families of
 GTRS codes are non-GRS.
 Referring to the results of \cite{G. Seroussi}
 and \cite{A. Lempel1}, we recognize that the systematic generator
 matrix $(I|A)$ of a GRS code has the property that $A$ is a Cauchy matrix. Considering that many MDS GTRS codes are non-GRS, it is natural to investigate the differences between the systematic generator matrices of GTRS and GRS codes. In this paper, we provide the systematic generator matrix for a class of GTRS codes and also propose a construction of non-GRS MDS codes based on this systematic generator matrix.

This paper is organized as follows.  Section 2
presents the necessary definitions, notions, and results.
In Section 3, we investigate the non-GRS property of two
special families of GTRS codes. Section 4 provides
the systematic generator matrix of a class of GTRS codes,
and in Section 5, we present our construction of non-GRS MDS codes.
Finally, Section 6 concludes our work.

\section{Preliminaries}
Let $\mathbb{F}_{q}$ be a finite field with $q$
 elements, where $q$
  is a power of a prime number.
  Denote the multiplicative group of $\mathbb{F}_{q}$ by $\mathbb{F}_{q}^{*}=\mathbb{F}_{q}\backslash{\{0\}}$,
  and let $\mathbb{F}_{q}[x]$ represent the polynomial
  ring over $\mathbb{F}_{q}$.
  We will summarize an important property of
  the systematic generator matrix of a linear MDS code.
\begin{Proposition}\label{mdspro}(\cite[Page 321]{R. Roth})
An $[n,k,d]$ code with the systematic generator matrix $(I_{k}|A_{k\times(n-k)})$
is MDS if and only if every square submatrix
(formed by any $i$ rows and any $i$ columns
for any $i=1,2,\dots, \min\{k,n-k\}$) of $A_{k\times(n-k)}$ is invertible.
\end{Proposition}

Below is the formal definition of generalized Reed-Solomon codes:
\begin{Definition}\label{def-1}(\cite[Page 323]{R. Roth})
Let $\alpha_{1},\dots,\alpha_{n}\in\mathbb{F}_{q}\bigcup\{\infty\}$ be
distinct elements with $k<n$  and $v_{1},\dots,v_{n}\in\mathbb{F}_{q}^{*}$.
The corresponding generalized Reed-Solomon (GRS) code is defined by
$$
GRS_{n,k}(\bm\alpha,\bm v)= \{(v_{1}f(\alpha_{1}),\dots,v_{n}f(\alpha_{n}))\,|\, f\in \mathbb{F}_{q}[x], \deg f<k\}.
$$
In this setting, for a polynomial $f(x)$ of degree $\deg f(x)<k$, the quantity $f(\infty)$ is defined as the coefficient of $x^{k-1}$ in the polynomial $f(x)$. In the case $\bm v=\bm 1=(1,1,\dots,1)$, the code is called a Reed-Solomon (RS) code.
\end{Definition}

It is well-known that
GRS codes are MDS \cite{R. Roth}.
In \cite{G. Seroussi,A. Lempel1}, the authors studied the systematic generator matrix of GRS codes using Cauchy matrices. An $m\times r$ matrix $A=(a_{ij})$ is called a Cauchy matrix if $a_{ij}=c_{i}d_{j}/(x_{i}+y_{j})$ for some elements $c_{1},\dots,c_{m},d_{1},\dots,d_{r},x_{1},\dots,x_{m},y_{1},\dots,y_{r}$ in $\mathbb{F}_{q}$ with $c_{i}d_{j}\neq 0$, $x_{i}$ distinct, $y_{i}$ distinct and $x_{i}+y_{j}\neq 0$ for all $i,j$. If $m=r$, then the determinant of $A$ is given by (see \cite[Page 323]{R. Roth})

\begin{equation}\label{eq2.1}
\det(A)=\prod_{1\leq i\leq m}c_{i}\prod_{1\leq j\leq m}d_{j}\frac{\prod_{1\leq i<j \leq m}(x_{j}-x_{i})(y_{j}-y_{i})}{\prod_{1\leq i,j\leq m}(x_{i}+y_{j})}.
\end{equation}

In \cite{G. Seroussi},
the authors introduced the concept of an
{\it extended Cauchy matrix}.
An extended Cauchy matrix
$\overline{A}$ contains a row (or a column) of the form
 $c_{\infty}(d_{1} \dots d_{r})$ (or $d_{\infty}(c_{1} \dots c_{m})$),
 and removing this row (or column) transforms
 $\overline{A}$ into the Cauchy matrix  $A=(a_{ij})$
 mentioned above.
They also provided a formula for the determinant of
an extended Cauchy matrix:
\begin{equation}\label{eq2.2}
\det(\overline{A})=(-1)^{m-h}\prod_{1\leq i\leq m}c_{i}\prod_{1\leq j\leq m}d_{j}\frac{\prod_{1\leq i<j \leq m-1}(x_{j}-x_{i})\prod_{1\leq i<j \leq m}(y_{j}-y_{i})}{\prod_{1\leq j\leq m,1\leq i \leq m-1}(x_{i}+y_{j})}
\end{equation}
where $h$ is the index of the extended row.


It was also established in \cite{G. Seroussi} that the matrix
 $(I_{k}|A_{n-k})$ generates a GRS code $GRS_{n,k}(\bm\alpha,\bm v)$ with no $\alpha_{i}=\infty$  if and only if $A_{n-k}$ is a Cauchy matrix.
 Furthermore,  $(I_{k}|A_{n-k})$ generates a GRS code $GRS_{n,k}(\bm\alpha,\bm v)$ with one of the $\alpha_{i}=\infty$  if and only if $A_{n-k}$ is a extended Cauchy matrix.

From the proof of \cite[Lemma 1]{A. Lempel1},
we can restate it as follows:
\begin{Lemma}\label{rothcauchy}
Given an $m\times r$ Cauchy or extended Cauchy matrix $A=(a_{ij})$ over $\mathbb{F}_{q}$, we can always assume $a_{1j}=d_{j}$ and $a_{2j}=d_{j}y_{j}^{-1}$ for $1\leq j\leq r$.
\end{Lemma}

Summarizing the main results from
\cite{G. Seroussi} and \cite{A. Lempel1}, we have:
\begin{Lemma}\label{GRScond}
An $[n,k]$ linear code over $\mathbb{F}_{q}$ with systematic generator matrix $(I_{k}|A)$ is a GRS code if and only if there exist nonzero elements $d_{1},\dots,d_{n-k},c_{1},\dots,c_{k-2}$, distinct nonzero $x_{1},\dots,x_{k-2}$ and distinct nonzero $y_{1},\dots,y_{n-k}$ in $\mathbb{F}_{q}$ with $x_{i}+y_{j}\neq0$ for all $i,j$ such that the first row of $A=(a_{ij})$ is $a_{1j}=d_{j}$ for $1\leq j\leq n-k$, the second row of $A=(a_{ij})$ is $a_{2j}=d_{j}y_{j}^{-1}$ for $1\leq j\leq n-k$, and the $i-$th row of $A=(a_{ij})$ is $a_{ij}=\frac{c_{i-2}d_{j}}{x_{i-2}+y_{j}}$ for all $3\leq i\leq k$ and $1\leq j\leq n-k$.
\end{Lemma}
Note that the matrix in Lemma \ref{GRScond}
is an extended Cauchy matrix with its first row
being the extended row. For convenience,
we will refer to the matrix in Lemma
\ref{GRScond} as a Cauchy matrix in this paper.

Next, we introduce the definition of generalized twisted
Reed-Solomon (GTRS) codes, drawing on the
definitions provided in \cite{P. Beelen, P. Beelen1}.
A {\it multiset} is a generalization of a set that allows
for element repetition. The {\it type number} of a multiset
is the count of distinct elements it contains.
The {\it multiplicity} of an element refers
to how many times it appears in the multiset.
A multiset is usually expressed by listing its
distinct elements along with their multiplicities.
For example, the multiset $\{a,b,c,a,c\}$
can be denoted by $\{2\cdot a,1\cdot b,2\cdot c\}$.
We say that a submultiset $M$ is included in a set $S$ if
 $x\in S$ for all $x\in M$ and the multiplicity of
 each elements in $M$ is less than or equal the multiplicity in $S$.

Let $k\leq n$ and $\ell\geq 1$ be positive integers. Let $\bm h=\{h_{1},h_{2},...,h_{\ell}\}\subseteq\{0,1,...,k-1\}$ where $\bm h$ is a multiset such that $\bm h=\{\ell_{1}\cdot h^{'}_{1},\ell_{2}\cdot h^{'}_{2},\dots,\ell_{s}\cdot h^{'}_{s}\}$ where $\sum_{i=1}^{s}\ell_{i}=\ell$ and $0\leq h^{'}_{1}<h^{'}_{2}<\dots<h^{'}_{s}\leq k-1$. Let $\bm t=\{t_{1},t_{2},...,t_{\ell}\}\subseteq\{1,...,n-k\}$ and $\bm \eta=\{\eta_{1},\eta_{2},...,\eta_{\ell}\}\subseteq\mathbb{F}_{q}^{*}$ where $\bm t$ and $\bm\eta$ are multisets.  The set of $[\bm t,\bm h,\bm \eta]$-{\it twisted polynomials} in $x$ is defined as

$$
\mathcal{P}_{k,n}[\bm t,\bm h,\bm\eta]
=\Big{\{}\sum_{i=0}^{k-1}f_{i}x^{i}+\sum_{j=1}^{\ell}\eta_{j}f_{h_{j}}x^{k-1+t_{j}}
\,|\,
f_{i}\in\mathbb{F}_{q},i=0,\dots,k-1\Big{\}}\subseteq\mathbb{F}_{q}[x].
$$

Let $\bm\alpha=(\alpha_{1},\alpha_{2},...,\alpha_{n})$ and $\bm v=(v_1,v_2,\dots, v_n)$ be two vectors of length $n$ over $\mathbb{F}_q$. We define the evaluation map related to $\bm\alpha$ and $\bm v$ as follows:

$$
ev_{\bm\alpha, \bm v}:\mathbb{F}_{q}[x]\rightarrow\mathbb{F}_{q}^{n},~f(x) \mapsto (v_1f(\alpha_{1}),v_2f(\alpha_{2}),...,v_nf(\alpha_{n})).
$$
The notion of  generalized twisted Reed-Solomon(GTRS)code is reproduced below, see \cite{WEIDONG}.
\begin{Definition}\label{def-2}
Let the entries of $\bm\alpha=(\alpha_{1},\alpha_{2},...,\alpha_{n})\in\mathbb{F}_{q}^{n}$ be pairwise distinct
and let  $\bm v=(v_{1},v_{2},...,v_{n})\in (\mathbb{F}_{q}^{*})^{n}$.
Let $\mathbf{t},\mathbf{h},\bm\eta,\ell$ and
$\mathcal{P}_{k,n}[\mathbf{t},\mathbf{h},\bm\eta]$ be defined as above. The generalized twisted Reed-Solomon(GTRS)code of length $n$, dimension $k$, twists $\ell$ and locators $\bm\alpha$ is defined by
$$
GTRS_{k,n}[\bm\alpha,\bm t,\bm h,\bm\eta,\bm v]=\{ev_{\bm\alpha, \bm v}(f)\,|\, f\in\mathcal{P}_{k,n}[\mathbf{t},\mathbf{h},\bm\eta]\}.
$$
\end{Definition}


\begin{Remark}{\rm
Clearly, $\mathcal{P}_{k,n}[\mathbf{t},\mathbf{h},\bm\eta]$ is a $k-$dimensional linear space and the map $ev_{\bm\alpha, \bm v}(f)$ is a linear map from $\mathcal{P}_{k,n}[\mathbf{t},\mathbf{h},\bm\eta]$ to $\mathbb{F}_{q}^{n}$. Since $1\leq t_{i}\leq n-k$ for all $i=1,\dots,\ell$, thus the degree of all the polynomials in $\mathcal{P}_{k,n}[\mathbf{t},\mathbf{h},\bm\eta]$ is less than $n$ which implies that $ev_{\bm\alpha, \bm v}(f)$ is injective, so we know the code in Definition \ref{def-2} is of dimension $k$. The condition on $\bm h$ ensures each $f_{h_{j}}$ is in the coefficients set $\{f_{0},f_{1},\dots,f_{k-1}\}$ of $f(x)$ for all $j=1,\dots,\ell$.}
\end{Remark}

A polynomial in $\mathcal{P}_{k,n}[\mathbf{t},\mathbf{h},\bm\eta]$ is of the form
$f(x) = \sum_{i=0}^{k-1}f_{i}x^{i}+\sum_{j=1}^{\ell}\eta_{j}f_{h_{j}}x^{k-1+t_{j}} $ and we denote by
$f(\bm\alpha)=(f(\alpha_{1}),f(\alpha_{2}),\dots,f(\alpha_{n}))$ its evaluation in ${\bm \alpha}$.
A generator matrix for the $GTRS_{n,k}[\bm\alpha,\bm t,\bm h,\bm\eta,\bm v]$ code is thus obtained by finding $k$ such polynomials $f^{(0)},\ldots,f^{(k-1)}$ whose evaluation vectors are linearly independent:
$$\left(
\begin{matrix}
f^{(0)}(\bm\alpha)\\
f^{(1)}(\bm\alpha)\\
\vdots\\
f^{(k-1)}(\bm\alpha)\\
\end{matrix}
\right)D_{\bm v}$$
where $D_{\bm v}$ is a diagonal matrix with the coefficients of ${\bf v}$ on its diagonal. Choose $f^{(l)}(x)$ such that $f_{l}=1$ and $f_{m}=0$ for $m\neq l$, that is
\[
f^{(l)}(x) = x^l +
\sum_{j,h_j=l}\eta_jx^{k-1+t_j}
\]
so $f^{(l)}(x) = x^l$ if there is no $j$ such that $h_j=l$.

Set $\bm\alpha^{i}=(\alpha_{1}^{i},\alpha_{2}^{i},\dots,\alpha_{n}^{i})$, we can write
$$
\left(
\begin{matrix}
f^{(0)}(\bm\alpha)\\
f^{(1)}(\bm\alpha)\\
\vdots\\
f^{(k-1)}(\bm\alpha)\\
\end{matrix}
\right)=
\left(
\begin{matrix}
\bm 1\\
\bm\alpha\\
\vdots\\
\bm\alpha^{k-1}
\end{matrix}
\right)+
\left(
\begin{matrix}
\bm e_{0}\\
\bm e_{1}\\
\vdots\\
\bm e_{k-1}\\
\end{matrix}
\right),
$$
where
${\bm e}_l = {\bf 0} $ when
$l \not\in \{h_{1},h_{2},\dots,h_{\ell}\}$, and
$
{\bm e}_l = \sum_{j,h_j=l}\eta_j{\bm\alpha}^{k-1+t_j}
$ else.

We call the above generator matrix the {\it standard generator matrix} of a GTRS code.

Specially, when $\ell=1$, $\bm t=(t),~\bm h=(h)$ and $\bm\eta=(\eta)$ where $1\leq t\leq n-k,~0\leq h\leq k-1$ and $\eta\neq0$. In this case, denote $\bm t,~\bm h$ and $\bm\eta$ as $t,~h$ and $\eta$, respectively. Then the standard generator matrix of $GTRS_{k,n}[\bm\alpha, t, h,\eta,\bm v]$ is given by
$$
\left(
\begin{matrix}
\bm 1\\
\bm\alpha\\
\vdots\\
\bm\alpha^{h-1}\\
\bm\alpha^{h}+\eta\bm\alpha^{k-1+t}\\
\bm\alpha^{h+1}\\
\vdots\\
\bm\alpha^{k-1}
\end{matrix}
\right).
$$

\section{Non-GRS property of GTRS codes}
In \cite{C.ZHU,P. Beelen1}, many GTRS codes have been proved to be non-GRS when $\ell=1$. For $(t,h)=(1,0),(1,k-1)$, GTRS codes are non-GRS when $3\leq k\leq \frac{n}{2}$ (see \cite{P. Beelen1}). However, it is unknown whether GTRS codes of above parameters are GRS or not when $\frac{n}{2}< k\leq n-3$.
In this section, by the
Schur product method, we study the non-GRS property of $GTRS_{k,n}[\bm\alpha, 1, k-1,\eta,\bm v]$ and $GTRS_{k,n}[\bm\alpha, 1, 0,\eta,\bm v]$ for $\frac{n}{2}<k\leq n-3$.

 By the Schur product method, for an $[n,k]$ linear code with $k\leq\frac{n}{2}$, if the dimension of its square code is not $2k-1$ then it is non-GRS. It is known that the dual codes of GRS codes are also GRS codes. If the dual codes of $GTRS_{k,n}[\bm\alpha, 1, k-1,\eta,\bm v]$ and $GTRS_{k,n}[\bm\alpha, 1, 0,\eta,\bm v]$ are non-GRS when $\frac{n}{2}<k\leq n-3$ then themselves are also non-GRS.

For $1\leq i\leq n$, suppose $u_{i}=\prod_{j=1,j\neq i}^{n}(\alpha_{i}-\alpha_{j})^{-1}$. For any integer $m$, suppose $\lambda_{m}=\sum_{i=1}^{n}u_{i}\alpha_{i}^{m}$, then $\lambda_{m}=0$ for $0\leq m\leq n-2$, and $\lambda_{n-1}=1,~\lambda_{-1}\neq 0$ (see \cite{WEIDONG,zhangjun}).

\begin{Proposition}\label{pro31}
The code $GTRS_{k,n}[\bm\alpha, 1, k-1,\eta,\bm v]$ is non-GRS for $3\leq k\leq n-3$.
\end{Proposition}
\begin{proof}
The case $3\leq k\leq\frac{n}{2}$ has been solved in \cite{P. Beelen1}. In the following, we suppose $k>\frac{n}{2}$ and prove the non-GRS property via its dual's non-GRS property.

By \cite[Theorem 5]{WEIDONG}, we know a parity check matrix of $GTRS_{k,n}[\bm\alpha, 1, k-1,\eta,\bm v]$ is
$$
H=\left(
\begin{matrix}
\bm 1\\
\bm\alpha\\
\vdots\\
\bm\alpha^{n-k-2}\\
\bm\alpha^{n-k}-\frac{1+\eta \lambda_{n}}{\eta}\bm\alpha^{n-k-1}
\end{matrix}
\right)D_{\frac{\bm u}{\bm v}}
$$
where $\frac{\bm u}{\bm v}=(\frac{u_1}{v_1},\dots,\frac{u_n}{v_n})$ and $D_{\frac{\bm u}{\bm v}}$ is a diagonal matrix with the nonzero coefficients of $\frac{\bm u}{\bm v}$ on its diagonal. Since the nonzero column multiple make no difference on the non-GRS property, without loss of generality, we assume
$$
H=\left(
\begin{matrix}
\bm 1\\
\bm\alpha\\
\vdots\\
\bm\alpha^{n-k-2}\\
\bm\alpha^{n-k}-\frac{1+\eta \lambda_{n}}{\eta}\bm\alpha^{n-k-1}
\end{matrix}
\right).
$$

Let $\bm h_{i}$ be the $i$th row of $H$, $D$ be the square code of $GTRS_{k,n}[\bm\alpha, 1, k-1,\eta,\bm v]^{\perp}$. First, since $\bm h_{i}*\bm h_{j}=\bm\alpha^{i+j-2}$ for all $1\leq i,j\leq n-k-1$, we know $D$ contains the vectors $\bm\alpha^{i}$ where $0\leq i\leq 2(n-k)-4$.

Then the vector $\bm\alpha^{2(n-k)-3}=\bm h_{n-k-2}*\bm h_{n-k}+\frac{1+\eta \lambda_{n}}{\eta}\bm\alpha^{2(n-k)-4}$ belongs to $D$.

Similarly, the vector $\bm\alpha^{2(n-k)-2}=\bm h_{n-k-1}*\bm h_{n-k}+\frac{1+\eta \lambda_{n}}{\eta}\bm\alpha^{2(n-k)-3}$ also belongs to $D$.

Finally, the vector $\bm\alpha^{2(n-k)}-2\frac{1+\eta \lambda_{n}}{\eta}\bm\alpha^{2(n-k)-1}=\bm h_{n-k}*\bm h_{n-k}-(\frac{1+\eta \lambda_{n}}{\eta})^{2}\bm\alpha^{2(n-k)-2}$ belongs to $D$.

Thus the dimension of $D$ is at least $2k$ which implies $GTRS_{k,n}[\bm\alpha, 1, k-1,\eta,\bm v]^{\perp}$ is non-GRS. Then $GTRS_{k,n}[\bm\alpha, 1, k-1,\eta,\bm v]$ is non-GRS.
\end{proof}

In the following, we study the non-GRS property of $GTRS_{k,n}[\bm\alpha, 1, 0,\eta,\bm v]$. First, we suppose all coefficients of $\bm\alpha$ and $\bm v$ are nonzero and $\eta\neq 0$, in \cite{zhangjun}, the authors introduced the code $\mathcal{C}_{k-1}(\bm\alpha,k,\eta,\bm v)$ which is generated by
$$
\left(
\begin{matrix}
\bm 1\\
\bm\alpha\\
\vdots\\
\bm\alpha^{k-2}\\
\bm\alpha^{k-1}+\eta\bm\alpha^{-1}
\end{matrix}
\right)D_{\bm v}.
$$

They proved that (see \cite[Corollary 1]{zhangjun}) $\mathcal{C}_{k-1}(\bm\alpha,k,\eta,\bm v)^{\perp}$ is generated by
$$
\left(
\begin{matrix}
\bm\alpha\\
\vdots\\
\bm\alpha^{n-k-1}\\
\bm\alpha^{n-k}+\eta^{'}\bm 1
\end{matrix}
\right)D_{\frac{\bm u}{\bm v}},
$$
where $\eta^{'}=-\frac{\lambda_{n-1}}{\eta\lambda_{-1}}=-\frac{1}{\eta\lambda_{-1}}\neq0$.

Obviously, $\mathcal{C}_{k-1}(\bm\alpha,k,\eta,\bm v)$ is also generated by
$$
\left(
\begin{matrix}
\bm 1\\
\bm\alpha\\
\vdots\\
\bm\alpha^{k-2}\\
\eta^{-1}\bm\alpha^{k-1}+\bm\alpha^{-1}
\end{matrix}
\right)D_{\bm v}.
$$

Since
$$
\left(
\begin{matrix}
\bm 1\\
\bm\alpha\\
\vdots\\
\bm\alpha^{k-2}\\
\eta^{-1}\bm\alpha^{k-1}+\bm\alpha^{-1}
\end{matrix}
\right)D_{\bm v}=
\left(
\begin{matrix}
\bm\alpha\\
\vdots\\
\bm\alpha^{k-2}\\
\bm\alpha^{k-1}\\
\bm 1+\eta^{-1}\bm\alpha^{k}\\
\end{matrix}
\right)D_{\bm v*\bm\alpha^{-1}},
$$
then $\mathcal{C}_{k-1}(\bm\alpha,k,\eta,\bm v)=GTRS_{k,n}[\bm\alpha, 1, 0,\eta^{-1},\bm v*\bm\alpha^{-1}]$.

\begin{Proposition}\label{pro32}
The code $GTRS_{k,n}[\bm\alpha, 1, 0,\eta,\bm v]$ is non-GRS for all $3\leq k\leq n-4$. Specially, it is non-GRS for all $3\leq k\leq n-3$ if all coefficients of $\bm\alpha$ are nonzero.
\end{Proposition}
\begin{proof}
The case $3\leq k\leq\frac{n}{2}$ has been solved in \cite{P. Beelen1}. First, suppose $k>\frac{n}{2}$ and all coefficients of $\bm\alpha$ are nonzero. Then $GTRS_{k,n}[\bm\alpha, 1, 0,\eta,\bm v]=\mathcal{C}_{k-1}(\bm\alpha,k,\eta^{-1},\bm v*\bm\alpha)$.

By~\cite[Corollary 1]{zhangjun}, $GTRS_{k,n}[\bm\alpha, 1, 0,\eta,\bm v]^{\perp}$ is generated by
$$
G=
\left(
\begin{matrix}
\bm\alpha\\
\vdots\\
\bm\alpha^{n-k-1}\\
\bm\alpha^{n-k}+\eta^{''}\bm 1
\end{matrix}
\right)D_{\frac{\bm u}{\bm v*\bm\alpha}},
$$
where $\eta^{''}=-\frac{1}{\eta\lambda_{-1}}\neq0$.

Similarly, we assume
$$
G=
\left(
\begin{matrix}
\bm\alpha\\
\vdots\\
\bm\alpha^{n-k-1}\\
\bm\alpha^{n-k}+\eta^{''}\bm 1
\end{matrix}
\right).
$$

Let $\bm g_i$ be the $i$th row of $G$, $B$ be the square code of $GTRS_{k,n}[\bm\alpha, 1, 0,\eta,\bm v]^{\perp}$. Since $\bm g_{i}*\bm g_{j}=\bm\alpha^{i+j}$ for all $1\leq i,j\leq n-k-1$, we know $B$ contains the vectors $\bm\alpha^{i}$ where $2\leq i\leq 2(n-k)-2$.

Thus the vector $\bm\alpha^{2(n-k)-1}=\bm g_{n-k-1}*\bm g_{n-k}-\eta^{''}\bm\alpha^{n-k-1}$ belongs to $B$.

Similarly, the vector $\bm\alpha=\frac{1}{\eta^{''}}(\bm g_{1}*\bm g_{n-k}-\bm\alpha^{n-k-1})$ also belongs to $B$.

Then the vector $\bm\alpha^{2(n-k)}+\eta^{''}\bm 1=\bm g_{n-k}*\bm g_{n-k}-2\eta^{''}\bm\alpha^{n-k}$ belongs to $B$.

Thus the dimension of $B$ is at least $2k$ which implies $GTRS_{k,n}[\bm\alpha, 1, 0,\eta,\bm v]^{\perp}$ is non-GRS. Then $GTRS_{k,n}[\bm\alpha, 1, 0,\eta,\bm v]$ is non-GRS for $\frac{n}{2}\leq k\leq n-3$ when all coefficients of $\bm\alpha$ are nonzero.

When $0\in\{\alpha_{1},\dots,\alpha_{n}\}$ and $\frac{n}{2}\leq k\leq n-4$, deleting the corresponding coordinate of $0$, we get a non-GRS code. Thus $GTRS_{k,n}[\bm\alpha, 1, 0,\eta,\bm v]$ is an extended code of a non-GRS code which implies itself is also non-GRS.
\end{proof}
\begin{Remark}{\rm
As we have referred, an $[n,n-2]$ MDS code is GRS since its codimension is $2$. In Proposition \ref{pro32}, we suppose $k\neq n-3$ when $0\in\{\alpha_{1},\dots,\alpha_{n}\}$. The reason is as following: after deleting a coordinate of an $[n,n-3]$ MDS code, we get an $[n-1,n-3]$ MDS code which is GRS since its codimension is $2$. Then we cannot apply the last statement in above proof for that case.}
\end{Remark}
\section{The systematic generator matrix of GTRS codes with $s=1$ }
In \cite{HUZ}, the authors gave a necessary and sufficient condition
for GTRS codes to be MDS; actually,
they also gave a condition for finding an information set for a GTRS code. In this section, we give a necessary and sufficient condition (Lemma \ref{infor})
for GTRS codes to be MDS, our condition is essentially same as the condition in \cite{HUZ}, but not same in the form. First, we fix some notations.


For any distinct elements $\alpha_{1},\alpha_{2},...,\alpha_{n}$ of $\mathbb{F}_{q}$ as
in Lemma \ref{infor} and any subset $\bm b=\{b_{1},\dots,b_{k}\}\subseteq\{1,\dots,n\}$,
consider the following linear equations

$$(f^{\bm b}_{0,i_{j}},f^{\bm b}_{1,i_{j}},\dots,f^{\bm b}_{k-1,i_{j}})
\left(
\begin{matrix}
1&1&\dots&1\\
\alpha_{b_{1}}&\alpha_{b_{2}}&\dots&\alpha_{b_{k}}\\
\vdots&\vdots&\ddots&\vdots\\
\alpha_{b_{1}}^{k-1}&\alpha_{b_{2}}^{k-1}&\dots&\alpha_{b_{k}}^{k-1}
\end{matrix}
\right)
=(\alpha_{b_{1}}^{k-1+t_{i_{j}}},\alpha_{b_{2}}^{k-1+t_{i_{j}}},\dots,\alpha_{b_{k}}^{k-1+t_{i_{j}}}),
$$
for $i=1,\dots,s$ and $j=1,\dots,\ell_{i}$. Since the Vandermonde matrix is invertible, the vector $(f^{\bm b}_{0,i_{j}},f^{\bm b}_{1,i_{j}},\dots,f^{\bm b}_{k-1,i_{j}})$ is unique for any $i,~j$. Set the polynomial $f^{\bm b}_{i_{j}}(x)=\sum_{m=0}^{k-1}f^{\bm b}_{m,i_{j}}x^{m}$ for any subsets $\bm b$ and any $i,~j$. In words, the coefficients of this polynomial are formed by a solution to the above system of linear equations.

\begin{Lemma}\label{infor}
Let $GTRS_{k,n}[\bm\alpha,\bm t,\bm h,\bm\eta,\bm v]$ be
defined as in Definition \ref{def-2}. Then
The coordinate set
$\bm b\subseteq\{1,\dots,n\}$ with $|\bm b|=k$
is an information set if and only if the following matrix is invertible:
$$
\left(
\begin{matrix}
1+\sum_{i=1}^{\ell_{1}}\eta_{1_{i}}f^{\bm b}_{h^{'}_{1},1_{i}}&\sum_{i=1}^{\ell_{1}}\eta_{1_{i}}f^{\bm b}_{h^{'}_{2},1_{i}}&\dots&\sum_{i=1}^{\ell_{1}}\eta_{1_{i}}f^{\bm b}_{h^{'}_{s},1_{i}}\\
\sum_{i=1}^{\ell_{2}}\eta_{2_{i}}f^{\bm b}_{h^{'}_{1},2_{i}}&1+\sum_{i=1}^{\ell_{2}}\eta_{2_{i}}f^{\bm b}_{h^{'}_{2},2_{i}}&\dots&\sum_{i=1}^{\ell_{2}}\eta_{2_{i}}f^{\bm b}_{h^{'}_{s},2_{i}}\\
\vdots&\vdots&\ddots&\vdots\\
\sum_{i=1}^{\ell_{s}}\eta_{s_{i}}f^{\bm b}_{h^{'}_{1},s_{i}}&\sum_{i=1}^{\ell_{s}}\eta_{s_{i}}f^{\bm b}_{h^{'}_{2},s_{i}}&\dots&1+\sum_{i=1}^{\ell_{s}}\eta_{s_{i}}f^{\bm b}_{h^{'}_{s},s_{i}}
\end{matrix}
\right)_{s\times s}.
$$
Specially, $GTRS_{k,n}[\bm\alpha,\bm t,\bm h,\bm\eta,\bm v]$ is MDS if and only if above matrix is invertible for any $\bm b\subseteq\{1,\dots,n\}$ with $|\bm b|=k$.
\end{Lemma}
\begin{proof}
For simple, we just prove the case $\bm v=(1,1,\dots,1)$, $\ell=2$, $h_{1}<h_{2}$ and $t_{1}\neq t_{2}$, the general case is totally same.

Since a $[n,k]$ linear code is MDS if and only if any $k$ columns of its generator matrix are linear independent. Thus $GTRS_{k,n}[\bm\alpha,\bm t,\bm h,\bm\eta,\bm v]$ is MDS if and only if for any $\bm b=\{b_{1},\dots,b_{k}\}\subseteq\{1,\dots,n\}$ the following matrix
\tiny
\begin{equation*}
\begin{blockarray}{ccccc}
&&&&\\
\begin{block}{(cccc)c}
1&1&\dots&1&\\
\alpha_{b_{1}}&\alpha_{b_{2}}&\dots&\alpha_{b_{k}}&\\
\vdots&\vdots&&\vdots&\\
\alpha_{b_{1}}^{h_{1}-1}&\alpha_{b_{2}}^{h_{1}-1}&\dots&\alpha_{b_{k}}^{h_{1}-1}&\\
\alpha_{b_{1}}^{h_{1}}+\eta_{1}\alpha_{b_{1}}^{k-1+t_{1}}&\alpha_{b_{2}}^{h_{1}}+\eta_{1}\alpha_{b_{2}}^{k-1+t_{1}}&\dots&\alpha_{b_{k}}^{h_{1}}+\eta_{1}\alpha_{b_{k}}^{k-1+t_{1}}&\cdots h_{1}+1~row\\
\alpha_{b_{1}}^{h_{1}+1}&\alpha_{b_{2}}^{h_{1}+1}&\dots&\alpha_{b_{k}}^{h_{1}+1}&\\
\vdots&\vdots&&\vdots&\\
\alpha_{b_{1}}^{h_{2}-1}&\alpha_{b_{2}}^{h_{2}-1}&\dots&\alpha_{b_{k}}^{h_{2}-1}&\\
\alpha_{b_{1}}^{h_{2}}+\eta_{2}\alpha_{b_{1}}^{k-1+t_{2}}&\alpha_{b_{2}}^{h_{2}}+\eta_{2}\alpha_{b_{2}}^{k-1+t_{2}}&\dots&\alpha_{b_{k}}^{h_{2}}+\eta_{2}\alpha_{b_{k}}^{k-1+t_{2}}&\cdots h_{2}+1~row\\
\alpha_{b_{1}}^{h_{2}+1}&\alpha_{b_{2}}^{h_{2}+1}&\dots&\alpha_{b_{k}}^{h_{2}+1}&\\
\vdots&\vdots&&\vdots&\\
\alpha_{b_{1}}^{k-1}&\alpha_{b_{2}}^{k-1}&\dots&\alpha_{b_{k}}^{k-1}&\\
\end{block}
\end{blockarray}
\end{equation*}
\normalsize
is invertible.

That is the matrix
\tiny
\begin{equation*}
G_{\bm b}=
\begin{blockarray}{ccccc}
&&&&\\
\begin{block}{(cccc)c}
1&1&\dots&1&\\
\alpha_{b_{1}}&\alpha_{b_{2}}&\dots&\alpha_{b_{k}}&\\
\vdots&\vdots&&\vdots&\\
\alpha_{b_{1}}^{h_{1}-1}&\alpha_{b_{2}}^{h_{1}-1}&\dots&\alpha_{b_{k}}^{h_{1}-1}&\\
\alpha_{b_{1}}^{h_{1}}+\eta_{1}f_{1}^{\bm b}(\alpha_{b_{1}})&\alpha_{b_{2}}^{h_{1}}+\eta_{1}f_{1}^{\bm b}(\alpha_{b_{2}})&\dots&\alpha_{b_{k}}^{h_{1}}+\eta_{1}f_{1}^{\bm b}(\alpha_{b_{k}})&\cdots h_{1}+1~row\\
\alpha_{b_{1}}^{h_{1}+1}&\alpha_{b_{2}}^{h_{1}+1}&\dots&\alpha_{b_{k}}^{h_{1}+1}&\\
\vdots&\vdots&&\vdots&\\
\alpha_{b_{1}}^{h_{2}-1}&\alpha_{b_{2}}^{h_{2}-1}&\dots&\alpha_{b_{k}}^{h_{2}-1}&\\
\alpha_{b_{1}}^{h_{2}}+\eta_{2}f_{2}^{\bm b}(\alpha_{b_{1}})&\alpha_{b_{2}}^{h_{2}}+\eta_{2}f_{2}^{\bm b}(\alpha_{b_{2}})&\dots&\alpha_{b_{k}}^{h_{2}}+\eta_{2}f_{2}^{\bm b}(\alpha_{b_{k}})&\cdots h_{2}+1~row\\
\alpha_{b_{1}}^{h_{2}+1}&\alpha_{b_{2}}^{h_{2}+1}&\dots&\alpha_{b_{k}}^{h_{2}+1}&\\
\vdots&\vdots&&\vdots&\\
\alpha_{b_{1}}^{k-1}&\alpha_{b_{2}}^{k-1}&\dots&\alpha_{b_{k}}^{k-1}&\\
\end{block}
\end{blockarray}
\end{equation*}
\normalsize
is invertible.

By elementary row transformation, we can eliminate all the elements $\sum_{j\neq h_{1},h_{2}}f_{j,1}(\alpha_{m})$ in $(h_{1}+1)-$th row and $\sum_{j\neq h_{1},h_{2}}f_{j,2}(\alpha_{m})$ in $(h_{2}+1)-$th row of $G_{\bm b}$. Specifically, let
\tiny
\begin{equation*}
P_{1}=
\begin{blockarray}{cccccccccc}
&&&(h_{1}+1)-th~column&&&(h_{2}+1)-th~column&&\\
&&&\vdots&&&\vdots&&\\
\begin{block}{(ccccccccc)c}
1&&&&&&&&\\
&\ddots&&&&&&&\\
&&1&&&&&&\\
-\eta_{1}f^{\bm b}_{0,1}&\dots&-\eta_{1}f^{\bm b}_{h_{1}-1,1}&1&-\eta_{1}f_{h_{1}+1,1}&\dots&0&\dots&-\eta_{1}f^{\bm b}_{k-1,1}&\cdots(h_{1}+1)-th~row\\
\vdots&&\vdots&\vdots&\vdots&&\vdots&&\\
-\eta_{2}f^{\bm b}_{0,2}&\dots&-\eta_{2}f^{\bm b}_{h_{1}-1,2}&0&-\eta_{2}f^{\bm b}_{h_{1}+1,2}&\dots&1&\dots&-\eta_{2}f^{\bm b}_{k-1,2}&\cdots(h_{2}+1)-th~row\\
&&&&&&&\ddots&\\
&&&&&&&&1\\
\end{block}
\end{blockarray}
\end{equation*}
\normalsize
then
\tiny
\begin{equation*}
P_{1}G_{\bm b}=
\begin{blockarray}{ccccc}
&&&&\\
\begin{block}{(cccc)c}
1&1&\dots&1&\\
\alpha_{b_{1}}&\alpha_{b_{2}}&\dots&\alpha_{b_{k}}&\\
\vdots&\vdots&&\vdots&\\
\alpha_{b_{1}}^{h_{1}}+\eta_{1}(f_{h_{1},1}^{\bm b}\alpha_{b_{1}}^{h_{1}}+f_{h_{2},1}^{\bm b}\alpha_{b_{1}}^{h_{2}})&\alpha_{b_{2}}^{h_{1}}+\eta_{1}(f_{h_{1},1}^{\bm b}\alpha_{b_{2}}^{h_{1}}+f_{h_{2},1}^{\bm b}\alpha_{b_{2}}^{h_{2}})&\dots&\alpha_{b_{k}}^{h_{1}}+\eta_{1}(f_{h_{1},1}^{\bm b}\alpha_{b_{k}}^{h_{1}}+f_{h_{2},1}^{\bm b}\alpha_{b_{k}}^{h_{2}})&\cdots h_{1}+1~row\\
\vdots&\vdots&&\vdots&\\
\alpha_{b_{1}}^{h_{2}}+\eta_{2}(f_{h_{1},2}^{\bm b}\alpha_{b_{1}}^{h_{1}}+f_{h_{2},2}^{\bm b}\alpha_{b_{1}}^{h_{2}})&\alpha_{b_{2}}^{h_{2}}+\eta_{2}(f_{h_{1},2}^{\bm b}\alpha_{b_{2}}^{h_{1}}+f_{h_{2},2}^{\bm b}\alpha_{b_{2}}^{h_{2}})&\dots&\alpha_{b_{k}}^{h_{2}}+\eta_{2}(f_{h_{1},2}^{\bm b}\alpha_{b_{k}}^{h_{1}}+f_{h_{2},2}^{\bm b}\alpha_{b_{k}}^{h_{2}})&\cdots h_{2}+1~row\\
\vdots&\vdots&&\vdots&\\
\alpha_{b_{1}}^{k-1}&\alpha_{b_{2}}^{k-1}&\dots&\alpha_{b_{k}}^{k-1}&\\
\end{block}
\end{blockarray}
\end{equation*}
\normalsize
is invertible.

Let
\tiny
\begin{equation*}
D=
\begin{blockarray}{ccccccccccc}
&&&(h_{1}+1)-th~column&&&&(h_{2}+1)-th~column&&&\\
&&&\vdots&&&&\vdots&&&\\
\begin{block}{(cccccccccc)c}
1&&&&&&&&&&\\
&\ddots&&&&&&&&\\
&&1&&&&&&&&\\
&&&1+\eta_{1}f^{\bm b}_{h_{1},1}&&&&\eta_{1}f^{\bm b}_{h_{2},1}&&&\cdots(h_{1}+1)-th~row\\
&&&&1&&&&&&\\
&&&&&\ddots&&&&&\\
&&&&&&1&&&&\\
&&&\eta_{2}f^{\bm b}_{h_{1},2}&&&&1+\eta_{2}f^{\bm b}_{h_{2},2}&&&\cdots(h_{2}+1)-th~row\\
&&&&&&&&\ddots&&\\
&&&&&&&&&1&\\
\end{block}
\end{blockarray}
\end{equation*}
\normalsize
and
$$
V_{\bm b}=
\left(
\begin{matrix}
1&1&\dots&1\\
\alpha_{b_{1}}&\alpha_{b_{2}}&\dots&\alpha_{b_{k}}\\
\vdots&\vdots&&\vdots\\
\alpha_{b_{1}}^{k-1}&\alpha_{b_{2}}^{k-1}&\dots&\alpha_{b_{k}}^{k-1}
\end{matrix}
\right),
$$
then we have $P_{1}G_{\bm b}=DV_{\bm b}$. Since $P_{1}$ and $V_{\bm b}$ are invertible, thus $G_{\bm b}$ is invertible if and only if $D$ is invertible if and only if the following matrix
$$
\left(
\begin{matrix}
1+\eta_{1}f^{\bm b}_{h_{1},1}&\eta_{1}f^{\bm b}_{h_{2},1}\\
\eta_{2}f^{\bm b}_{h_{1},2}&1+\eta_{2}f^{\bm b}_{h_{2},2}\\
\end{matrix}
\right)
$$
is invertible. This completes the proof.
\end{proof}

In the following, suppose $s=1$, that is $h=h_{1}=h_{2}=...=h_{\ell}$. Since any $k$ columns of the standard generator matrix $G$ of a GTRS code have the same form, so we assume the first $k$ columns of the standard generator matrix is linearly independent i.e. $\{1,2,\dots,k\}$ is an information set.
For convenience, when $\bm b=\{1,2,\dots,k\}$,
we denote $(f^{\bm b}_{0,i_{j}},f^{\bm b}_{1,i_{j}},\dots,f^{\bm b}_{k-1,i_{j}})$ and $f^{\bm b}_{i_{j}}(x)=\sum_{m=0}^{k-1}f^{\bm b}_{m,i_{j}}x^{m}$ as $(f_{0,i_{j}},f_{1,i_{j}},\dots,f_{k-1,i_{j}})$ and
$f_{i_{j}}(x)=\sum_{m=0}^{k-1}f_{m,i_{j}}x^{m}$ respectively. By Lemma \ref{infor}, what we assume is $1+\sum_{i=1}^{\ell}\eta_{i}f_{h,i}\neq0$.

\begin{Proposition}\label{systemati}
Let $GTRS_{k,n}[\bm\alpha,\bm t,\bm h,\bm\eta,\bm v]$ be defined as in Definition \ref{def-2} with $s=1$ and assume the first $k$ columns of the generator matrix is linearly independent.
Then the systematic generator matrix of $GTRS_{k,n}[\bm\alpha,\bm t,\bm h,\bm\eta,\bm v]$ is $[I_{K}~|~\bm C_{\bm\alpha,\bm v}+\bm r\bm d]$ where
\begin{equation*}
\begin{split}
\bm C_{\bm\alpha,\bm v}=&\left(
\begin{matrix}
v_{1}^{-1}&&\\
&\ddots&\\
&&v_{k}^{-1}\\
\end{matrix}
\right)
\left(
\begin{matrix}
1&1&\dots&1\\
\alpha_{1}&\alpha_{2}&\dots&\alpha_{k}\\
\vdots&\vdots&\ddots&\vdots\\
\alpha_{1}^{k-1}&\alpha_{2}^{k-1}&\dots&\alpha_{k}^{k-1}
\end{matrix}
\right)^{-1}\cdot\\
&\left(
\begin{matrix}
1&1&\dots&1\\
\alpha_{k+1}&\alpha_{k+2}&\dots&\alpha_{n}\\
\vdots&\vdots&\ddots&\vdots\\
\alpha_{k+1}^{k-1}&\alpha_{k+2}^{k-1}&\dots&\alpha_{n}^{k-1}
\end{matrix}
\right)
\left(
\begin{matrix}
v_{k+1}&&\\
&\ddots&\\
&&v_{n}\\
\end{matrix}
\right)
\end{split}
\end{equation*}
is the corresponding Cauchy matrix of $GRS_{n,k}(\bm\alpha,\bm v)$, the column vector $\bm r$ is the $(h+1)-$th column of
$$
\left(
\begin{matrix}
v_{1}^{-1}&&\\
&\ddots&\\
&&v_{k}^{-1}\\
\end{matrix}
\right)
\left(
\begin{matrix}
1&1&\dots&1\\
\alpha_{1}&\alpha_{2}&\dots&\alpha_{k}\\
\vdots&\vdots&\ddots&\vdots\\
\alpha_{1}^{k-1}&\alpha_{2}^{k-1}&\dots&\alpha_{k}^{k-1}
\end{matrix}
\right)^{-1}
$$
and the row vector
\begin{equation*}
\begin{split}
&\bm d=[(1+\sum_{i=1}^{\ell}\eta_{i}f_{h,i})^{-1}(\sum_{i=1}^{\ell}\eta_{i}\alpha_{k+1}^{k-1+t_{i}}-\sum_{i=1}^{\ell}\eta_{i}f_{i}(\alpha_{k+1})),\dots,\\
&(1+\sum_{i=1}^{\ell}\eta_{i}f_{h,i})^{-1}(\sum_{i=1}^{\ell}\eta_{i}\alpha_{n}^{k-1+t_{i}}-\sum_{i=1}^{\ell}\eta_{i}f_{i}(\alpha_{n}))]
\left(
\begin{matrix}
v_{k+1}&&\\
&\ddots&\\
&&v_{n}\\
\end{matrix}
\right).
\end{split}
\end{equation*}
\normalsize
\end{Proposition}
\begin{proof}
We first prove the case $\bm v=\bm 1$. Denote by $G=(G_{k}|G_{n-k})$
where the matrix $G_{k}$ is formed by the first $k$ columns of $G$ and
$G_{n-k}$ the last $n-k$ columns of $G$.

Let $f_{i}(x)=\sum_{j=0}^{k-1}f_{j,i}x^{j}$ for all $i=1,\dots,\ell$, then
$$
G_{k}=
\left(
\begin{matrix}
1&1&\dots&1\\
\alpha_{1}&\alpha_{2}&\dots&\alpha_{k}\\
\vdots&\vdots&\ddots&\vdots\\
\alpha_{1}^{h}+\sum_{i=1}^{\ell}\eta_{i}f_{i}(\alpha_{1})&\alpha_{2}^{h}+\sum_{i=1}^{\ell}\eta_{i}f_{i}(\alpha_{2})&\dots&\alpha_{k}^{h}+\sum_{i=1}^{\ell}\eta_{i}f_{i}(\alpha_{k})\\
\vdots&\vdots&\ddots&\vdots\\
\alpha_{1}^{k-1}&\alpha_{2}^{k-1}&\dots&\alpha_{k}^{k-1}
\end{matrix}
\right).
$$

By elementary row transformation, we can eliminate all the elements of the form $\sum^{k-1}_{j\neq h,j=0}f_{j,i}(\alpha_{m})$ in $(h+1)-$th row of $G_{k}$. Specifically, let
\tiny
\begin{equation*}
P_{1}=
\begin{blockarray}{cccccccc}
&&&the~(h+1)-th~columon&&&\\
&&&\vdots&&&\\
\begin{block}{(ccccccc)c}
1&&&&&&\\
&\ddots&&&&&\\
&&1&&&&\\
-\sum_{i=1}^{\ell}\eta_{i}f_{0,i}&\dots&-\sum_{i=1}^{\ell}\eta_{i}f_{h-1,i}&1&-\sum_{i=1}^{\ell}\eta_{i}f_{h+1,i}&
\dots&-\sum_{i=1}^{\ell}\eta_{i}f_{k-1,i}&\cdots the ~(h+1)-th~row\\
&&&&1&&\\
&&&&&\ddots&\\
&&&&&&1\\
\end{block}
\end{blockarray}
\end{equation*}
\normalsize

and then
$$
P_{1}G_{k}=
\left(
\begin{matrix}
1&1&\dots&1\\
\alpha_{1}&\alpha_{2}&\dots&\alpha_{k}\\
\vdots&\vdots&\ddots&\vdots\\
\alpha_{1}^{h}+\sum_{i=1}^{\ell}\eta_{i}f_{h,i}\alpha_{1}^{h}&\alpha_{2}^{h}+\sum_{i=1}^{\ell}\eta_{i}f_{h,i}\alpha_{2}^{h}&\dots&\alpha_{k}^{h}+\sum_{i=1}^{\ell}\eta_{i}f_{h,i}\alpha_{k}^{h}\\
\vdots&\vdots&&\vdots\\
\alpha_{1}^{k-1}&\alpha_{2}^{k-1}&\dots&\alpha_{k}^{k-1}
\end{matrix}
\right).
$$
Let
\begin{equation*}
P_{2}=
\begin{blockarray}{cccccc}
&&&&\\
\begin{block}{(ccccc)c}
1&&&&\\
&\ddots&&&\\
&&(1+\sum_{i=1}^{\ell}\eta_{i}f_{h,i})^{-1}&&&\cdots the ~(h+1)-th~row\\
&&&\ddots&&\\
&&&&1&\\
\end{block}
\end{blockarray}
\end{equation*}
then
$$
P_{2}P_{1}G_{k}=
\left(
\begin{matrix}
1&1&\dots&1\\
\alpha_{1}&\alpha_{2}&\dots&\alpha_{k}\\
\vdots&\vdots&\ddots&\vdots\\
\alpha_{1}^{h}&\alpha_{2}^{h}&\dots&\alpha_{k}^{h}\\
\vdots&\vdots&\ddots&\vdots\\
\alpha_{1}^{k-1}&\alpha_{2}^{k-1}&\dots&\alpha_{k}^{k-1}
\end{matrix}
\right).
$$
Let
$$
V=
\left(
\begin{matrix}
1&1&\dots&1\\
\alpha_{1}&\alpha_{2}&\dots&\alpha_{k}\\
\vdots&\vdots&\ddots&\vdots\\
\alpha_{1}^{k-1}&\alpha_{2}^{k-1}&\dots&\alpha_{k}^{k-1}
\end{matrix}
\right),
$$
then we have
$$V^{-1}P_{2}P_{1}G=V^{-1}P_{2}P_{1}(G_{k}|G_{n-k})=(I_{k}|V^{-1}P_{2}P_{1}G_{n-k}).$$
By direct calculation, we have
$P_{2}P_{1}G_{n-k}=(\bm a_{1},\dots,\bm a_{n-k})$ where
$$
\bm a_{m}=\left(
\begin{matrix}
1\\
\alpha_{k+m}\\
\vdots\\
(1+\sum_{i=1}^{\ell}\eta_{i}f_{h,i})^{-1}(\alpha_{k+m}^{h}+\sum_{i=1}^{\ell}\eta_{i}\alpha_{k+m}^{k-1+t_{i}}-\sum_{i=1}^{\ell}\eta_{i}\sum_{j\neq h}f_{j,i}\alpha_{k+m}^{j})\\
\vdots\\
\alpha_{k+m}^{k-1}\\
\end{matrix}
\right)
$$

$$
=\left(
\begin{matrix}
1\\
\alpha_{k+m}\\
\vdots\\
\alpha_{k+m}^{h}\\
\vdots\\
\alpha_{k+m}^{k-1}\\
\end{matrix}
\right)+
\left(
\begin{matrix}
0\\
0\\
\vdots\\
(1+\sum_{i=1}^{\ell}\eta_{i}f_{h,i})^{-1}(\sum_{i=1}^{\ell}\eta_{i}\alpha_{k+m}^{k-1+t_{i}}-\sum_{i=1}^{\ell}\eta_{i}f_{i}(\alpha_{k+m}))\\
\vdots\\
0\\
\end{matrix}
\right)
$$
for $m=1,\dots,n-k$.
Thus, $V^{-1}P_{2}P_{1}G_{n-k}=C_{\bm\alpha,\bm 1}+\bm r\bm d$ where

$$
C_{\bm\alpha,\bm 1}=
\left(
\begin{matrix}
1&1&\dots&1\\
\alpha_{1}&\alpha_{2}&\dots&\alpha_{k}\\
\vdots&\vdots&\ddots&\vdots\\
\alpha_{1}^{k-1}&\alpha_{2}^{k-1}&\dots&\alpha_{k}^{k-1}
\end{matrix}
\right)^{-1}
\left(
\begin{matrix}
1&1&\dots&1\\
\alpha_{k+1}&\alpha_{k+2}&\dots&\alpha_{n}\\
\vdots&\vdots&\ddots&\vdots\\
\alpha_{k+1}^{k-1}&\alpha_{k+2}^{k-1}&\dots&\alpha_{n}^{k-1}
\end{matrix}
\right)
$$
is the corresponding Cauchy matrix of $GRS_{n,k}(\bm\alpha,\bm 1)$, the column vector $\bm r$ is the $(h+1)-$th column of
$$
\left(
\begin{matrix}
1&1&\dots&1\\
\alpha_{1}&\alpha_{2}&\dots&\alpha_{k}\\
\vdots&\vdots&\ddots&\vdots\\
\alpha_{1}^{k-1}&\alpha_{2}^{k-1}&\dots&\alpha_{k}^{k-1}
\end{matrix}
\right)^{-1}
$$
and the row vector
\begin{equation*}
\begin{split}
\bm d&=
[(1+\sum_{i=1}^{\ell}\eta_{i}f_{h,i})^{-1}(\sum_{i=1}^{\ell}\eta_{i}\alpha_{k+1}^{k-1+t_{i}}-\sum_{i=1}^{\ell}\eta_{i}f_{i}(\alpha_{k+1})),\dots,\\
&(1+\sum_{i=1}^{\ell}\eta_{i}f_{h,i})^{-1}(\sum_{i=1}^{\ell}\eta_{i}\alpha_{n}^{k-1+t_{i}}
-\sum_{i=1}^{\ell}\eta_{i}f_{i}(\alpha_{n}))].
\end{split}
\end{equation*}
\normalsize

For $\bm v=(v_{1},\dots,v_{n})\neq\bm 1$, let $\mathbf{diag}\{v_{1},\dots,v_{n}\}$ denote  the diagonal matrix with the diagonal elements being $v_{1},\dots,v_{n}$. It follows that $$
G\mathbf{diag}\{v_{1},\dots,v_{n}\}=(G_{k}|G_{n-k})\mathbf{diag}\{v_{1},\dots,v_{n}\}
$$ generates $GTRS_{k,n}[\bm\alpha,\bm t,\bm h,\bm\eta,\bm v]$. Multiply it from the left hand side by $\mathbf{diag}\{v^{-1}_{1},\dots,v^{-1}_{k}\}G^{-1}_{k}$, we get

$\mathbf{diag}\{v^{-1}_{1},\dots,v^{-1}_{k}\}G^{-1}_{k}(G_{k}|G_{n-k})\mathbf{diag}\{v_{1},\dots,v_{n}\}$
\begin{eqnarray*}
&=&\mathbf{diag}\{v^{-1}_{1},\dots,v^{-1}_{k}\}(I_{k}|G^{-1}_{k}G_{n-k})\mathbf{diag}\{v_{1},\dots,v_{n}\}\\
&=&\mathbf{diag}\{v^{-1}_{1},\dots,v^{-1}_{k}\}(\mathbf{diag}\{v_{1},\dots,v_{k}\}|G^{-1}_{k}G_{n-k}\mathbf{diag}\{v_{k+1},\dots,v_{n}\})\\
&=&(I_{k}|\mathbf{diag}\{v^{-1}_{1},\dots,v^{-1}_{k}\}G^{-1}_{k}G_{n-k}\mathbf{diag}\{v_{k+1},\dots,v_{n}\})
\end{eqnarray*}
which completes the proof.
%
\end{proof}

\begin{Example}{\rm
Let $\ell=1$, $t_{1}=1$, $h_{1}=0$ and $\bm v=\bm 1$. Let
$q$ be an odd prime power,  $H=\{\alpha_{1},\alpha_{2},\dots,\alpha_{2n}\}$ be a subgroup of $\mathbb{F}_{q}^{*}$ and $\{\alpha_{1},\alpha_{2},\dots,\alpha_{n}\}$ be the subgroup of $H$ with order $n$.
Let $(-1)^{n}\eta_{1}^{-1}\in\mathbb{F}^{*}_{q}\setminus H$ and consider the $[2n,n]$ GTRS code $GTRS_{n,2n}[\bm\alpha,1,0,\eta_{1},\bm 1]$.

By the above conditions, we know $(f_{0,1},f_{1,1},\dots,f_{n-1,1})=(1,0,\dots,0)$ and
$$
\left(
\begin{matrix}
1&1&\dots&1\\
\alpha_{1}&\alpha_{2}&\dots&\alpha_{n}\\
\vdots&\vdots&\ddots&\vdots\\
\alpha_{1}^{n-1}&\alpha_{2}^{n-1}&\dots&\alpha_{n}^{n-1}
\end{matrix}
\right)^{-1}=\frac{1}{n}
\left(
\begin{matrix}
1&\alpha^{-1}_{1}&\dots&\alpha^{-(n-1)}_{1}\\
1&\alpha^{-1}_{2}&\dots&\alpha^{-(n-1)}_{2}\\
\vdots&\vdots&\ddots&\vdots\\
1&\alpha_{n}^{-1}&\dots&\alpha_{n}^{-(n-1)}
\end{matrix}
\right).
$$
Then we have $\bm r=\left(
\begin{aligned}
\frac{1}{n}\\
\vdots~\\
\frac{1}{n}\\
\end{aligned}
\right)
$
and $\bm d=\frac{\eta_{1}}{1+\eta_{1}}(\alpha_{n+1}^{n}-1,\dots,\alpha_{2n}^{n}-1)$.
By Proposition \ref{systemati}, we know the systematic generator matrix of $GTRS_{n,2n}[\bm\alpha,1,0,\eta_{1},\bm 1]$ is $(I_{n}|A_{n\times n})$ where $\bm C_{\bm\alpha,\bm 1}$ is a Cauchy matrix and
$$
A_{n\times n}=\bm C_{\bm\alpha,\bm 1}+\frac{\eta_{1}}{n(1+\eta_{1})}\left(
\begin{matrix}
\alpha_{n+1}^{n}-1&\alpha_{n+2}^{n}-1&\dots&\alpha_{2n}^{n}-1\\
\alpha_{n+1}^{n}-1&\alpha_{n+2}^{n}-1&\dots&\alpha_{2n}^{n}-1\\
\vdots&\vdots&\ddots&\vdots\\
\alpha_{n+1}^{n}-1&\alpha_{n+2}^{n}-1&\dots&\alpha_{2n}^{n}-1\\
\end{matrix}
\right).
$$
}
\end{Example}

%
%
%
%
%

\section{A construction of non-GRS MDS codes}
It is known that $[n,1,n]$ and $[n,n-1,2]$ MDS codes are trivial;
 $[n,2,n-1]$ and $[n,n-2,3]$ MDS codes are GRS codes.
 Thus, we assume that $3\leq k\leq n-3$ in this section.

 First choose some elements from $\mathbb{F}_{q}$; Let $d_{1},\dots,d_{n-k},c_{1},\dots,c_{k-2}$ be nonzero,
 $x_{1},\dots,x_{k-2}$ be nonzero and distinct and $y_{1},\dots,y_{n-k}$ be nonzero and distinct with $x_{i}+y_{j}\neq0$ for all $i,j$. Suppose the first row of $A=(a_{ij})$ is $a_{1j}=d_{j}$, the second row of $A=(a_{ij})$ is $a_{2j}=d_{j}y_{j}^{-1}$ for $1\leq j\leq n-k$, and the $i$th row of $A=(a_{ij})$ is $a_{ij}=\frac{c_{i-2}d_{j}}{x_{i-2}+y_{j}}$ for all $3\leq i\leq k$ and $1\leq j\leq n-k$.

 Let $\mathbb{F}_{q^{d}}\supseteq\mathbb{F}_{q}$ be an extension field where $d>1$. Let $E_{ij}$ denote the matrix whose $(i,j)$th entry is $1$ and others are zeros. Suppose $\beta\in\mathbb{F}_{q^{d}}\setminus\mathbb{F}_{q}$, build a new $k\times(n-k)$ matrix $D=A+\beta E_{11}$ over $\mathbb{F}_{q^{d}}$, i.e.

$$
D=
\left(
\begin{matrix}
d_{1}+\beta&d_{2}&\dots&d_{n-k}\\
d_{1}y_{1}^{-1}&d_{2}y_{2}^{-1}&\dots &d_{n-k}y_{n-k}^{-1}\\
\frac{c_{1}d_{1}}{x_{1}+y_{1}}&\frac{c_{1}d_{2}}{x_{1}+y_{2}}&\dots&\frac{c_{1}d_{n-k}}{x_{1}+y_{n-k}}\\
\vdots&\vdots&\ddots&\vdots\\
\frac{c_{k-2}d_{1}}{x_{k-2}+y_{1}}&\frac{c_{k-2}d_{2}}{x_{k-2}+y_{2}}&\dots&\frac{c_{k-2}d_{n-k}}{x_{k-2}+y_{n-k}}\\
\end{matrix}
\right)_{k\times(n-k)}.
$$

\begin{Theorem}\label{NONGRSMDS}
With the notations above, let $n,k$ be integers with $3\leq k\leq n-3$. Suppose the $[n,k]$ linear code $\mathcal{D}$ over $\mathbb{F}_{q^{d}}$ is generated by $(I_{k}|D)$, then $\mathcal{D}$ is an $[n,k]$ non-GRS MDS code.
\end{Theorem}
\begin{proof}
We prove that $\mathcal{D}$ is MDS via Proposition \ref{mdspro}. Since any square submatrix which does not contain the element $d_{1}+\beta$ is invertible (from Equations \ref{eq2.1} and \ref{eq2.2}, also from Proposition \ref{mdspro}). Thus we just consider the square submatrix which contains $d_{1}+\beta$.

Let $D\left(
\begin{aligned}
i_{1}&i_{2}&\dots&i_{m}\\
j_{1}&j_{2}&\dots&j_{m}\\
\end{aligned}
\right)
$
denotes the square submatrix of $D$ formed by $i_{1}<i_{2}<\dots<i_{m}$ rows and $j_{1}<j_{2}<\dots<j_{m}$ columns of $D$ where $1\leq m\leq\min\{k,n-k\}$ (we also use the above symbol for the matrix $A$). We just need to consider the case $i_{1}=j_{1}=1$.

Obviously,
$$D\left(
\begin{aligned}
i_{1}&\dots&i_{m}\\
j_{1}&\dots&j_{m}\\
\end{aligned}
\right)=A\left(
\begin{aligned}
i_{1}&\dots&i_{m}\\
j_{1}&\dots&j_{m}\\
\end{aligned}
\right)+\beta A\left(
\begin{aligned}
i_{2}&\dots&i_{m}\\
j_{2}&\dots&j_{m}\\
\end{aligned}
\right),$$
and the two minors from $A$ are nonzero over $\mathbb{F}_{q}$. Since $\beta\notin\mathbb{F}_{q}$, then $D\left(
\begin{aligned}
i_{1}&\dots&i_{m}\\
j_{1}&\dots&j_{m}\\
\end{aligned}
\right)\neq0$.
By Proposition \ref{mdspro}, $\mathcal{D}$ is MDS.

Next, we prove $\mathcal{D}$ is non-GRS. By Lemma \ref{GRScond}, if $\mathcal{D}$ is GRS, then there exist nonzero elements $d^{'}_{1},\dots,d^{'}_{n-k},c^{'}_{1},\dots,c^{'}_{k-2}$, distinct nonzeros $x^{'}_{1},\dots,x^{'}_{k-2}$ and distinct nonzeros $y^{'}_{1},\dots,y^{'}_{n-k}$ with $x^{'}_{i}+y^{'}_{j}\neq0$ for all $i,j$ such that the first row of $D$ is $(d^{'}_{1}d^{'}_{2}\dots d^{'}_{n-k})$, the second row of $D$ is $(\frac{d^{'}_{1}}{y^{'}_{1}}\frac{d^{'}_{2}}{y^{'}_{2}}\dots \frac{d^{'}_{n-k}}{y^{'}_{n-k}})$ and the $i-$th row of $D$ is $(\frac{c^{'}_{i-2}d^{'}_{1}}{x^{'}_{i-2}+y^{'}_{1}}\frac{c^{'}_{i-2}d^{'}_{2}}{x^{'}_{i-2}+y^{'}_{2}}\dots\frac{c^{'}_{i-2}d^{'}_{n-k}}{x^{'}_{i-2}+y^{'}_{n-k}})$ for all $i=3,\dots,k$.

From $d_{2}=d^{'}_{2}$ and $\frac{d_{2}}{y_{2}}=\frac{d^{'}_{2}}{y^{'}_{2}}$, we have $y_{2}=y^{'}_{2}$. Similarly, we also have $d_{3}=d^{'}_{3}$ and $y_{3}=y^{'}_{3}$. Then by $\frac{c_{1}d_{2}}{x_{1}+y_{2}}=\frac{c^{'}_{1}d^{'}_{2}}{x^{'}_{1}+y^{'}_{2}}$, $d_{2}=d^{'}_{2}$ and $y_{2}=y^{'}_{2}$, we have $\frac{c_{1}}{x_{1}+y_{2}}=\frac{c^{'}_{1}}{x^{'}_{1}+y_{2}}$. Similarly, we have $\frac{c_{1}}{x_{1}+y_{3}}=\frac{c^{'}_{1}}{x^{'}_{1}+y_{3}}$. Then we can get linear equations about $c^{'}_{1}$ and $x^{'}_{1}$ as

$$
\left\{
\begin{aligned}
(x_{1}+y_{2})c^{'}_{1}-c_{1}x^{'}_{1}=c_{1}y_{2}\\
(x_{1}+y_{3})c^{'}_{1}-c_{1}x^{'}_{1}=c_{1}y_{3}\\
\end{aligned}
\right..
$$

Since
$$
\det\left(
\begin{matrix}
x_{1}+y_{2}&-c_{1}\\
x_{1}+y_{3}&-c_{1}\\
\end{matrix}
\right)=
-c_{1}(y_{2}-y_{3})\neq0
$$
we know this linear equations has a unique solution.
Moreover,
$\left\{
\begin{aligned}
c^{'}_{1}=c_{1}\\
x^{'}_{1}=x_{1}\\
\end{aligned}
\right.
$ is clearly a solution of this linear equations, thus $c^{'}_{1}=c_{1}$ and $x^{'}_{1}=x_{1}$.

From $d_{1}+\beta=d^{'}_{1}$ and $\frac{d_{1}}{y_{1}}=\frac{d^{'}_{1}}{y^{'}_{1}}$, we have $y^{'}_{1}=\frac{y_{1}(d_{1}+\beta)}{d_{1}}$. Then by $\frac{c^{'}_{1}d^{'}_{1}}{x^{'}_{1}+y^{'}_{1}}=\frac{c_{1}d_{1}}{x_{1}+y_{1}}$ and
$\left\{
\begin{aligned}
c^{'}_{1}=c_{1}\\
x^{'}_{1}=x_{1}\\
\end{aligned}
\right.
$
we can get $\beta=0$ a contradiction, showing $\mathcal{D}$ is non-GRS.
\end{proof}

\begin{Example}{\rm
Let $q=7^{2}$, $\mathbb{F}_{7^{2}}=\mathbb{F}_{7}(\theta)$ with $\theta^{2}+2=0$. Let $\bm\alpha=(\infty,0,1,\dots,6)$ and $\bm v=\bm 1$, then $GRS_{8,3}(\bm\alpha,\bm 1)$ has the systematic generator matrix $(I_{3}|A_{3\times 5})$ where
$$
A_{3\times 5}=
\left(
\begin{matrix}
2&6&5&6&2\\
6&5&4&4&2\\
2&3&4&5&6\\
\end{matrix}
\right).
$$
Let
$$
D_{3\times 5}=
\left(
\begin{matrix}
2+\theta&6&5&6&2\\
6&5&4&4&2\\
2&3&4&5&6\\
\end{matrix}
\right)
$$
then $(I_{3}|D_{3\times 5})$ generates a $[8,3]$ non-GRS MDS code over $\mathbb{F}_{7^{2}}$.}
\end{Example}

\section{Conclusion}
This paper studies the non-GRS property of GTRS codes, gives the systematic generator matrices for a class of GTRS codes and a construction of non-GRS MDS codes. This construction is gotten by ''breaking" a Cauchy matrix. Different from the Schur product method, we prove the non-GRS property of our MDS codes by Cauchy matrix method.

\vskip 4mm

\noindent {\bf Acknowledgment.} This work was done while the first author was visiting the Division of Mathematical Sciences, School of Physical and Mathematical Sciences, Nanyang Technological University, Singapore. The first author is grateful for the hospitality and support, and sincerely thanks Professor Fr\'{e}d\'{e}rique Oggier for her hosting. The authors thank Fr\'{e}d\'{e}rique Oggier for her very helpful comments which have improved the presentation of this paper.  This work was supported by NSFC (Grant Nos. 12441102, 12271199) and China Scholarship Council (No. 202306770055).


\end{document}